\documentclass[a4paper]{amsart}
\usepackage{amssymb}
\usepackage{mathrsfs} 
\usepackage{tikz}
\newtheorem{theorem}{Theorem}[section]
\newtheorem{lemma}[theorem]{Lemma}
\newtheorem{cor}[theorem]{Corollary}
\newtheorem{prop}[theorem]{Proposition}
\theoremstyle{definition}
\newtheorem{definition}[theorem]{Definition}

\theoremstyle{remark}
\newtheorem{remark}[theorem]{Remark}

\numberwithin{equation}{section}
\newcommand{\dontprint}[1]{\relax}
\begin{document}

\title{Regularization of divergent integrals} 
\author{Giovanni Felder}
\address{Department of mathematics,
ETH Zurich, 8092 Zurich, Switzerland}
\email{felder@math.ethz.ch}
\author{David Kazhdan}
\address{Einstein Institute of Mathematics,
The Hebrew University of Jerusalem,
Jerusalem 91904, Israel}
\email{kazhdan@math.huji.ac.il}


\dedicatory{Dedicated to Sasha Beilinson the occasion of his 60th
  birthday}

\begin{abstract} We study the Hadamard finite part of divergent
  integrals of differential forms with singularities on
  submanifolds. We give formulae for the dependence of the finite part
  on the choice of regularization and express them in terms of a
  suitable local residue map. The cases where the submanifold is a
  complex hypersurface in a complex manifold and where it is a
  boundary component of a manifold with boundary, arising in string
  perturbation theory, are treated in more detail.
\end{abstract}

\maketitle
\tableofcontents
\section{Introduction}
Trying to make sense of integrals that are not absolutely convergent
is an old endeavour in mathematics and physics, that, despite its
apparent meaninglessness, has been surprisingly fruitful and useful in
many subjects. Hadamard defined the {\em finite part} of a divergent
integral by first introducing a cutoff, namely by integrating over the
complement of a small neighbourhood of the singular set of the
integrand, say of size $\epsilon$, and then letting $\epsilon$ tend to
zero, after subtraction of divergent terms, see \cite{Hadamard}, Book
III, Chapter I.  His motivation was to give a meaning to formal
solutions of differential equations, integrals that would be solutions
if one were allowed to differentiate under the integral sign. This
sort of questions as well as the related questions on the asymptotic
behaviour of level set integrals, were also one of the motivations for
the theory of generalized functions, see \cite{GelfandShilov}, Section
II.4. \dontprint{There the close relation between the asymptotic behaviour of
level set integrals of a function and the meromorphic dependence on
the exponent of integrals of a complex power of the function was recognized
and was one of the starting points of the theory of $D$-modules, see
\cite{Bernstein}.}

The same strategy of subtracting infinities from divergent integrals
was followed by Feynman to make sense of divergent integrals in
perturbative quantum field theory. Renormalised integrals are obtained
from divergent integrals by first regularizing them, by introducing a
cutoff as above or by a similar procedure, and then removing the
cutoff after subtracting divergent terms, see, e.g., Chapter 8 of
\cite{ItzyksonZuber}. In this context the question of dependence on
the choice of cutoff arises, and and the goal is to show that the
final results for physical quantities such as scattering amplitudes,
are independent of the choice of regularization. Since the Feynman
integrals are integrals over a Euclidean space, it is natural to
restrict to regularizations that use this structure, such as cutting
off points at distance $<\!\epsilon$ from the integration region. The
question becomes more subtle in perturbative string theory, where
Feynman integrals are replaced by integrals over moduli spaces of
curves. They can be interpreted as integrals of singular differential
forms on the Deligne--Mumford compactification of moduli spaces with
singularities on the compactification divisor. Again the integrals are
defined by cutting off a small neighbourhood of the divisor and study
the asymptotic behaviour as the size of the neighbourhood tends to
zero. Since there is no natural way to choose the family of shrinking
neighbourhoods, the question of dependence on the regularization is
subtle. In fact, in the case of superstrings the limit as the size
goes to zero exists without subtracting divergent terms, but it
depends on the regularization in a calculable way, showing that
changes of regularization can be absorbed into redefinition of the
coupling constants, see \cite{Witten}, Section 7.

Inspired by these calculations in superstring theory, in
\cite{FelderKazhdan2016} we considered integrals of products of
holomorphic and antiholomorphic differential forms on complex
manifolds with poles on hypersurfaces. In the case where the
antiholomorphic form has a simple pole, we gave formulae for the
dependence on the choice of cutoff function, generalizing a
calculation of \cite{Witten}. To treat the general case it is useful
to consider a more general setting, which is the approach of this
paper.

We consider integrals of differential forms on an oriented
$n$-dimensional manifold $X$ that are singular on a submanifold
$Y$. The kind of singularities we allow are determined by a {\em
  conformal class of nonnegative Morse--Bott functions} vanishing on
$Y$: a nonnegative Morse--Bott function with zero set $Y$ is a
nonnegative function $\mu$ on $X$ vanishing exactly on $Y$ with
non-degenerate Hessian in the normal direction. Given such a
Morse--Bott function $\mu$ we consider the space $\mathcal A_\mu(X)$
of differential forms $\omega$ on $X\smallsetminus Y$ such that, for
some integer $N\geq0$, $\mu^N\omega$ extends smoothly to $X$. Clearly
$\mathcal A_\mu(X)=\mathcal A_{f\mu}(X)$ for any positive smooth
function $f$ on $X$, so that only the conformal class of $\mu$, consisting of
all $f\mu$ with $f$ everywhere positive, plays a role. An important
special case is when $Y$ is a hypersurface in a complex manifold $X$,
the setting of \cite{FelderKazhdan2016}. In this case we have a
canonical conformal class of Morse--Bott functions, consisting of
nonnegative Morse--Bott functions locally divisible by $|f|^2$ for any
holomorphic function $f$ with a simple zero on $Y$.

Returning to the general case, we wish to give a meaning to the
divergent integral $\int_X\omega$ of a top differential form
$\omega\in\mathcal A_\mu^{\mathrm{dim}\,X}(X)$ whose support has compact closure in
$X$.  For small $\epsilon>0$, the inequality $\mu<\epsilon^2$ defines
a tubular neighbourhood of $Y$ and the integral over its complement is
well-defined.  It is then not difficult to see that, as
$\epsilon\to 0$,
\[
  \int_{\mu\geq\epsilon^2}\omega=\sum_{k=1}^{2N-m}I_{-k}\epsilon^{-k}+
  I_0\log\frac1\epsilon+I_{\mathrm{finite}}+O(\epsilon),
\]
where $m$ is the codimension of the submanifold $Y\subset X$. The {\em
  Hadamard finite part} of the divergent integral $\int_X\omega$ is
then by definition
$I_{\mathrm{finite}}=I_{\mathrm{finite}}(\mu,\omega)$. In general it
depends on the choice of nonnegative Morse--Bott function vanishing on
$Y$ and the question is to describe the dependence.

For this purpose it is useful to introduce the {\em zeta function}
$\zeta(s;\mu,\omega)$ defined as the meromorphic continuation of the
absolutely convergent integral
\[
  \zeta(s;\mu,\omega)=\int_X\mu^{s/2}\omega,\quad \mathrm{Re}\,s\gg 0.
\]
It turns out that the zeta function, as a function of $s$, has only
simple poles and that
$I_0=\mathrm{res}_{s=0}\zeta(s;\mu,\omega)$ is independent of
$\mu$ within its conformal class, see Theorem \ref{t-01}.  The first
result expresses the finite part and describes its dependence on the
Morse--Bott function in its conformal class in terms of the zeta
function and its residue at $0$.
\begin{theorem}\label{ti-1}
  Let $\omega\in\mathcal A_\mu^n(X)$,
  $\psi\in\mathcal A_\mu^{n-1}(X)$.
  \begin{enumerate}
  \item[(i)]$ I_{\mathrm{finite}}(\mu,\omega) =\lim_{s\to
      0}\left(\zeta(s;\mu,\omega)-\frac{I_0(\mu,\omega)}{s}\right).$
  \item[(ii)] For any smooth function $\varphi$,
    $I_{\mathrm{finite}}(e^{2\varphi}\mu,\omega)=
    I_{\mathrm{finite}}(\mu,\omega)+ I_0(\mu,\omega\varphi).$
  \item[(iii)]
    $I_{\mathrm{finite}}(\mu,d\psi)=\frac12I_0(\mu,\psi\wedge d\mu/\mu).$
  \end{enumerate}
\end{theorem}
Part (i) of this Theorem is proved in Section \ref{ss-2.3}, see
Theorem \ref{t-01}. Part (ii) is discussed in Section \ref{ss-2.4},
see Theorem \ref{t-02}. It extends the result of
\cite{FelderKazhdan2016}, where the case of complex hypersurfaces was
studied. Finally Part (iii) is proved in Section \ref{ss-3.1},
Proposition \ref{pr-02}.

We see that a key role is played by the map $I_0$.  It turns out that
$I_0$ vanishes if the codimension $m$ is odd, so in that case the
finite part is independent of the choice of Morse--Bott function
within a conformal class. Moreover, because of (iii), the finite part
is a well defined function on the cohomology of the complex
$\mathcal A_\mu(X)$. Thus this story is mostly interesting if $m=2r$
is even. In this case we derive a local formula for $I_0$ in terms of
a residue map.  For this it is useful to extend the setting and
consider differential forms $\omega\in\mathcal A_\mu$ not necessarily
of top degree. As we show in Theorem \ref{t-03}, the linear form
\[
  I_0(\mu,\omega\wedge\mbox{---})\colon\varphi\mapsto
  I_0(\mu,\omega\wedge\varphi)
\]
on smooth compactly supported forms of complementary degree, defines a
de Rham current with support in $Y$ and the map
$\omega\mapsto I_0(\mu,\omega\wedge\mbox{---})$ is a morphism of complexes
$\mathcal A_\mu(X)\to\mathcal D'(X)$ to the complex of currents.  We
then observe that $\mathcal A_\mu(X)$ is the algebra of global
sections of a sheaf $\mathcal A_{X,\mu}$ of differential graded
algebras. We show that $\mathcal A_{X,\mu}$ has a quasi-isomorphic
subcomplex $\mathcal A_{X,\mu}^{\mathrm{tame}}$ of differential forms
with tame singularities. By definition, $\omega\in\mathcal A_{X,\mu}$
has tame singularities if $\mu^r\omega$ and
$\mu^{r-1}d\mu\wedge\omega$ extend to smooth forms on $X$.

\begin{theorem} Let $m=2r$ be even and $i\colon Y\to X$ denote the
  inclusion map.  Assume also that both $X$ and $Y$ are oriented.
  There is a morphism of complexes of sheaves
  $R\colon\mathcal A^{\mathrm{tame}}_{X,\mu}\to i_*\mathcal A_Y[-m]$
  such that for any global differential form $\omega$ with tame
  singularities, and compactly supported smooth form $\varphi$ of
  complementary degree,
  \[
    I_0(\mu,\omega\wedge\varphi)=\int_YR(\omega)\wedge\varphi.
  \]
\end{theorem}
We prove this result in Section \ref{s-3}, see Theorem \ref{t-04}. The
orientability of $Y$, assumed here for simplicity of exposition, is
not really needed and is dropped there at the cost of involving the
orientation bundle of $Y$. The ``residue map'' $R$ can be given a
fairly explicit formula, see Theorem \ref{t-05}.  We then address the
question of comparing $\mathcal A_{X,\mu}$ for $\mu$ belonging to
different conformal classes. We show that the sheaves of differential
graded algebras $\mathcal A_{X,\mu}$ are essentially independent of
the conformal class of $\mu$: they come with a
quasi-isomorphism---unique up to a contractible space of choices---to
a homotopy colimit over the category of simplices of the singular set
of the cone of Morse--Bott functions, see Theorem \ref{t-Ania}.

In the last two Sections we focus on the important special cases where
$Y$ has codimension 2 and 1 in $X$, respectively.

In Section \ref{s-4} we specialize our results to the case of complex
hypersurfaces. The complex structure gives rise to a canonical class
of conformal structure and a canonical orientation of both $X$ and
$Y$. We recover and generalize results of \cite{FelderKazhdan2016} to
the case of arbitrary order of poles.  We also extend the results to
the case where $Y$ is a divisor with normal crossings: Theorem
\ref{ti-1} has a natural generalization, see Theorem \ref{t-06}, which
is however combinatorially slightly more involved.

In Section \ref{s-5} we treat the case of codimension 1. As mentioned
above, the odd codimension case is less involved as the zeta function
is regular at $s=0$. However in this case it is natural to extend our
setting and consider $Y$ to be the boundary of an oriented manifold
with boundary $X$. Then we have only one conformal class of
Morse--Bott functions (defined as squares of functions vanishing to
first order on the boundary) and a canonical orientation of $Y$. It
turns out that the zeta function has a pole at zero and that all our
results in the even codimension case have an analogue in the case of
manifolds with boundary, see Theorems \ref{t-b01}, \ref{t-b04}.

In the Appendix we calculate the cohomology sheaf of $\mathcal
A_\mu$. This calculation is used in Section \ref{s-3} to show that
$\mathcal A_\mu$ is quasi isomorphic to the subcomplex of differential
forms with tame singularities.
\subsection*{Acknowledgments} We thank Tomer M. Schlank and Yakov
Eliashberg for suggestions and explanations. The research of G. F. was
partially supported by the National Competence Centre in Research
SwissMAP---The mathematics of physics of the Swiss National Science
Foundation.
 
\section{Divergent integrals}\label{s-2}
\subsection{Nonnegative Morse--Bott functions}\label{ss-2.1}
Let $X$ be an $n$-dimensional smooth oriented manifold and
$Y\subset X$ a closed submanifold of dimension $n-m$. We consider
regularizations of divergent integrals $\int_X\omega$ of smooth
differential $n$-forms $\omega$ on $X\smallsetminus Y$. To do this,
following Hadamard, we cut out a small neighbourhood of $Y$ from the
integration and study the behaviour of the integral as the size of the
neighbourhood tends to zero. The neighbourhoods we cut out are
parametrized by a class of smooth functions that we now introduce.

\begin{definition} A {\em nonnegative Morse--Bott function} on $X$
  with zero set $Y$ is a smooth function
  $\mu\colon X\to\mathbb R_{\geq0}$ vanishing exactly on $Y$ and such
  that the rank of the Hessian on $Y$ is equal to the codimension of
  $Y$.  Two nonnegative Morse--Bott functions are called {\em
    conformally equivalent} if they vanish on the same submanifold $Y$
  and their ratio is an everywhere positive function. The equivalence
  classes for this equivalence relations are called {\em conformal
    classes}.
\end{definition}

By the Morse--Bott lemma, for each nonnegative Morse--Bott function
$\mu$ vanishing on $Y$ and point $x\in Y$ there are coordinate
functions $x_1,\dots, x_m$ on some neighbourhood of $x$ in $X$ such
that $Y$ is given there by the equations $x_1=\cdots=x_m=0$ and such
that $\mu=x_1^2+\cdots+x_m^2$.

Let $\mu$ be a nonnegative Morse--Bott function vanishing on $Y$. We
denote by $\mathcal A_\mu(X)$ the de Rham complex of differential
forms $\omega$ on $X\smallsetminus Y$ such that, for some integer $N$,
$\mu^N\omega$ extends to a smooth differential form on $X$.  Clearly
$\mathcal A_\mu(X)$ only depends on the conformal class of $\mu$.

\subsection{Zeta functions and level set integrals}\label{ss-2.2}
To an $n$-form $\omega\in\mathcal A^n_\mu(X)$ whose support has
compact closure in $X$ we associate the zeta function
\[
  \zeta(s;\mu,\omega)=\int_X\mu^{\frac s2}\omega,
\]
and to an $n-1$-form $\alpha\in\mathcal A_\mu^{n-1}(X)$ the level set
integral
\[
  I(t;\mu,\alpha)=\int_{\mu=t^2}\alpha,\quad t>0.
\]
The zeta function is defined and holomorphic for sufficiently large
$\mathrm{Re}\,s$.  It is related to the level set integral by a Mellin
transform:
\begin{lemma}
  Let $\alpha\in\mathcal A^{n-1}_\mu(X)$ with support in a
  sufficiently small neighbourhood of $Y$. Set
  \[
    \omega=\frac{d\mu}{2\mu}\wedge\alpha.
  \]
  Then
  \[
    \zeta(s;\mu,\omega)=\int_{0}^\infty t^{s-1} I(t;\mu,\alpha)dt
  \] The regularized integral is
  \[
    \int_{\mu\geq\epsilon^2}\omega=\int_\epsilon^\infty
    I(t;\mu,\alpha)\frac{dt}t.
  \]
\end{lemma}

\begin{proof}
  The assumption on the support allows us choose spherical coordinates
  on the support of $\alpha$ such that $r=\sqrt{\mu}$ is a radial
  coordinate. By a partition of unity argument we may assume that
  $\alpha$ has support in a coordinate neighbourhood of a point of $Y$
  and that $\mu=x_1^2+\cdots+x_m^2$ in a suitable coordinate
  system. In spherical coordinates $r>0,y\in S^{m-1}$ in the normal
  direction, we can replace $X\smallsetminus Y$ by
  $\mathbb R_{>0}\times S^{m-1}\times \mathbb R^{n-m}$ and write
  \begin{align*}
    \omega&=f\,dx_1\wedge\cdots\wedge dx_n\\
          &=f(r,y,x_{m+1},\cdots,x_{n})r^{m-1}dr
            \wedge d\Omega\wedge dx_{m+1}\wedge\cdots\wedge dx_{n}.
  \end{align*}
  Here $d\Omega$ is the volume form on the unit sphere. The zeta
  function can be then evaluated in polar coordinates:
  \[
    \zeta(s,\mu,\omega)=\int_0^\infty r^{s-1} J(r)dr,
  \]
  with
  \[
    J(r)=\int_{S^{m-1}\times \mathbb R^{n-m}}f(r,y,x)
    r^{m}d\Omega\wedge dx_{m+1}\wedge\cdots\wedge dx_n.
  \]
  The integrand is the coordinate expression of $\alpha$, so
  $J(r)=I(r;\mu,\alpha)$. The same calculation with $s=0$ and
  integration range $(\epsilon,\infty)$ gives the formula for the
  regularized integral.
\end{proof}
\begin{lemma} Assume that $\omega$ has support in a sufficiently small
  neighbourhood of $Y$ and that $\mu^N\omega$ extends to a smooth
  form on $X$. Then $t^{2N-m}I(t;\mu,\alpha) $ extends to a smooth,
  compactly supported, even function of $t\in\mathbb R$.
\end{lemma}
\begin{proof} By a partition of unity we may assume that the support
  of $\omega$ is contained in a small neighbourhood of a point of
  $Y$. By the Morse--Bott lemma we may also assume that there are
  local coordinates on that neighbourhood so that $Y$ is given by
  $x_{1}=\cdots=x_m=0$ and
  \[
    \mu=x_1^2+\cdots + x_m^2.
  \]
  Then $\omega=\mu^{-N}f\,dx_1\dots dx_m$ for some smooth function
  $f$. We use spherical coordinates in the normal direction: locally
  $U$ is $\mathbb R_{>0}\times S^{m-1}\times\mathbb R^{n-m}$, with
  radial coordinate $r=\sqrt\mu$, and
  \[
    \omega=f\, r^{-2N+m-1}dr\wedge d\Omega\wedge
    dx_{m+1}\wedge\cdots\wedge dx_n.
  \]
  Here $d\Omega$ is the volume form on the unit sphere. We can
  therefore choose $\alpha$ to be
  \[
    \alpha=f r^{-2N+m}d\Omega\wedge dx_{m+1}\wedge\cdots\wedge dx_n,
  \]
  and
  \[
    I(t;\mu,\alpha)=t^{-2N+m}\int_{S^{m-1}\times \mathbb
      R^{n-m}}f(ty,x)d\Omega(y)\wedge dx_{m+1}\wedge\cdots\wedge dx_n.
  \]
  It is clear that the integral on the right-hand side is defined for
  all $t\in\mathbb R$ and is a smooth compactly supported function of
  $t$. Since the involution $y\mapsto -y$ of the sphere maps $d\Omega$
  to $(-1)^{m}d\Omega$ and preserves/reverses the orientation if $m$
  is even/odd, we get
  \[
    I(-t;\mu,\alpha)=(-1)^m I(t;\mu,\alpha).
  \]
\end{proof}
\subsection{Hadamard finite part}\label{ss-2.3}
Here we consider the regularization of the divergent integral of a
differential forms in $\mathcal A_\mu(X)$ for some nonnegative Morse--Bott function
$\mu$ on $X$ vanishing on a submanifold $Y$ and define its Hadamard
finite part. While $\mathcal A_\mu(X)$ depends only on the
conformal class of $\mu$, the finite part depends on $\mu$ and we
describe its dependence within the conformal class.
\begin{theorem}\label{t-01} Let $\mu$ be a
  nonnegative Morse--Bott function vanishing on $Y$.  Let $\omega$ be
  an $n$-form on $X\smallsetminus Y$ such that $\mu^{N}\omega$ extends
  to a smooth form on $X$ with compact support.  Then
  \begin{enumerate}
  \item[(i)] $\int_X\mu^{s/2}\omega$ is holomorphic for
    $\mathrm{Re}\,s>2N-m$ and has a meromorphic continuation
    $\zeta(s;\mu,\omega)$ with at most simple poles on the arithmetic
    progression $s=2N-m,2N-m-2,\dots$.
  \item[(ii)] As $\epsilon\to 0$ we have an expansion
    \[
      \int_{\mu\geq\epsilon^2}\omega
      =\sum_{k=1}^{2N-m}I_{-k}(\mu,\omega)\epsilon^{-k}+
      I_0(\mu,\omega)\log\frac1\epsilon
      +I_{\mathrm{finite}}(\mu,\omega)+O(\epsilon)
    \] and $I_k=0$ unless $k\equiv m\mod 2$.
  \item[(iii)] For $k=1,\dots, 2N-m$, $k\equiv m \mod 2$,
    \[
      I_{-k}(\mu,\omega)
      =\frac1k\mathrm{res}_{s=k}\zeta(s;\mu,\omega),
    \]
    and
    \[
      I_0(\mu,\omega) =\mathrm{res}_{s=0}\zeta(s;\mu,\omega)
    \]
    is independent of $\mu$ within its conformal class; it vanishes if
    $m$ is odd. The finite part is
    \[
      I_{\mathrm{finite}}(\mu,\omega) =\lim_{s\to
        0}\left(\zeta(s;\mu,\omega)-\frac{I_0(\mu,\omega)}{s}\right).
    \]
  \end{enumerate}
\end{theorem}
\begin{proof}
  We may assume that $\omega$ has support in an arbitrary small
  neighbourhood of $Y$ since we can achieve this by adding to $\omega$
  a form on $X$ with support disjoint from $Y$.  Let
  \[
    I(t;\mu,\alpha)=\sum_{k=-2N+m}^p b_k t^k +R_p(t)
  \]
  be the Laurent expansion of $I$. Since $t^{2N-m}I$ extends to a
  smooth even function,
  \[
    b_k=0,\quad\text{if $k\not\equiv m\mod 2$}.
  \]
  For any $M>0$, the remainder $R_p(t)$ is bounded by $Ct^{p+1}$ for
  $t\in[0,M]$.  For sufficiently large $M$ and Re($s$) we have
  \begin{align}
    \zeta(s;\mu,\omega)&=\int_0^\infty t^{s-1}I(t;\mu,\alpha)dt \notag\\
                       &=\int_0^Mt^{s-1}I(t;\mu,\alpha)dt\\
                       &=\sum_{k=-2N+m}^pb_k\frac {M^{k+s}}{k+s}
                         +\int_0^Mt^{s-1}R_p(t)dt\notag
  \end{align}
  The last term is holomorphic for $\mathrm{Re}(s)>-p-1$. This proves
  (i) and gives the formula
  \[
    b_k=\mathrm{res}_{s=-k}\zeta(s;\mu,\omega)
  \]
  for the residues at the poles.  A similar calculation can be done
  for the integral with cutoff if we expand the integral on level sets
  up to order $p=0$ giving the proof of (ii):
  \begin{align*}
    \int_\epsilon^\infty I(t;\mu,\alpha)dt
    =&\int_\epsilon^M I(t;\mu,\alpha)dt\\
    =&\sum_{k=-2N+p}^{-1} b_k 
       \left( 
       \frac{M^k}k 
       - \frac{\epsilon^k}k 
       \right) 
       + b_0\log\frac M\epsilon
    \\
     &+\int_\epsilon^MR_0(t)\frac{dt}t
  \end{align*}
  The last integral is absolutely convergent for $\epsilon=0$ since
  $R_0(t)\leq Ct$. The coefficient of $\epsilon^{-k}$ is $b_{-k}/k$
  and is thus the residue of $\zeta$ at $s=k$ divided by $k$, as
  claimed in (iii). Similarly, the coefficient of the
  $\log(1/\epsilon)$ is $b_0=\mathrm{res}_{s=0}\zeta(s;\mu,\omega)$.
 
  To show the independence of $I_0$ on $\mu$ we write the ratio of two
  Morse--Bott functions as $\exp(2\varphi)$ for some smooth function
  $\varphi$.  For $s$ with large positive real part we have
  \begin{equation}\label{e-92}
    \int_X(e^{2\varphi}\mu)^{\frac s2}\omega-\int_X\mu^{\frac s2}\omega
    =s\int_{X}\mu^{\frac s2}\omega(s)
  \end{equation}
  where
  \[
    \omega(s)=\frac{e^{s\varphi}-1}{s}\,\omega
  \]
  is an entire function of $s$ and has a convergent expansion at $s=0$
  with coefficients in $\mathcal A_\mu(X)$.  Thus
  $\int_X\mu^{s/2}\omega(s)$ has an analytic continuation with at most
  a simple pole at $s=0$ and the right-hand side of \eqref{e-92} is
  regular there.  It follows that $\zeta(s;e^{2\varphi}\mu,\omega)$
  and $\zeta(s;\mu,\omega)$ have the same residue at $s=0$.

  Finally, the finite part is
  \[
    I_{\mathrm{finite}}(\mu,\omega)=\sum_{k=-2N+m}^{-1}b_k\,\frac{M^k}k+b_0\log
    M +\int_0^MR_0(t)\frac{dt}t,
  \]
  for any sufficiently large $M$.  It coincides with the value at
  $s=0$ of the expression above for $\zeta(s;\mu,\omega)$ after
  subtraction of the pole $I_0/s=b_0/s$ (the logarithmic term comes
  from $\lim_{s\to 0}(M^s/s-1/s)=\log M$).
\end{proof}
\begin{cor} The coefficient $I_{-k}(\mu,\omega)$, $k=1,2,\dots$ of
  $\epsilon^{-k}$ in $\int_{\mu\geq\epsilon^2}\omega$ is homogeneous
  of degree $k/2$ as a function of $\mu$.
\end{cor}
\begin{proof}
  This follows from the obvious identity
  $\zeta(s;t\mu,\omega)=t^{s/2}\zeta(s;\mu,\omega)$, $t>0$, and
  Theorem \ref{t-01} (iii). It is also clear from Theorem \ref{t-01}
  (ii): Replacing $\mu$ by $t\mu$ is the same as replacing $\epsilon$
  by $t^{-1/2}\epsilon$.
\end{proof}
\begin{definition} The {\em Hadamard finite part} of the divergent
  integral $\int_X\omega$ with Morse--Bott function $\mu$ is
  $I_{\mathrm{finite}}(\mu,\omega)$, as defined in Theorem \ref{t-01},
  (ii) or (iii).
\end{definition}

\subsection{Dependence on the Morse--Bott function}\label{ss-2.4}
The following result says how the finite part depends on the
Morse--Bott function in a conformal class.
\begin{theorem}\label{t-02}
  For any function $\varphi\in C^\infty(X)$,
  \[
    I_{\mathrm{finite}}(\mu
    e^{2\varphi},\omega)=I_{\mathrm{finite}}(\mu,\omega) +
    I_0(\mu,\varphi\,\omega),
  \]
  where
  $I_0(\mu,\varphi\,\omega)=
  \mathrm{res}_{s=0}\zeta(s;\mu,\varphi\,\omega)$ is independent of
  $\mu$.
\end{theorem}
\begin{proof}
  Let, as in the proof of Theorem \ref{t-01},
  \[
    \omega(s)=\frac{e^{s\varphi}-1}{s}\,\omega
  \]
  Since $\omega(0)=\varphi\,\omega$, we have by \eqref{e-92}
  \begin{align*}
    I_{\mathrm{finite}}(\mu e^{2\varphi},\omega)-I_{\mathrm{finite}}(\mu,\omega)
    &=\lim_{s\to0}(\zeta(s;e^{2\varphi}\mu,\omega)-\zeta(s;\mu,\omega))\\
    &=\mathrm{res}_{s=0}\zeta(s;\mu,\varphi\,\omega).
  \end{align*}
  By Theorem \ref{t-01} (iii), the right-hand side is independent of
  $\mu$.
\end{proof}
\begin{cor} Suppose that $\mu^{N}\omega$ extend to a smooth form on
  $X$ and that $\varphi \leq \mathrm{const} \, \mu^{N-(m-1)/2}$. Then
  \[
    I_{\mathrm{finite}}(\mu
    e^{2\varphi},\omega)=I_{\mathrm{finite}}(\mu,\omega)
  \]
\end{cor}
\begin{proof} In this case $\varphi\,\omega= f \mathrm{vol}$ for some
  smooth volume form vol and a function $f$ such that
  $|f|\leq \mathrm{const}\,\mu^{-(m-1)/2}$ which is integrable. Thus
  $\zeta(s;\mu,\varphi\,\omega)$ is smooth and given by its absolutely
  convergent integral representation at $s=0$.
\end{proof}

\begin{remark} If the codimension $m$ of $Y$ is odd, then $I_0=0$ and
  the finite part is independent of the choice of $\mu$ in its
  conformal class.
\end{remark}
\section{The de Rham complex of differential forms with
  singularities}\label{s-3}
Let $\mu$ be a nonnegative Morse--Bott function on an $n$-dimensional
manifold $X$ vanishing on a submanifold $Y$ of codimension $m$. Let
$\omega$ be a top degree differential on $X\smallsetminus Y$ such that
$\mu^{N}\omega$ extends to a form on $X$ with compact support for some
$N$.  The residue at zero of the zeta function
\[
  I_0(\mu,\omega)=\mathrm{res}_{s=0}\zeta(s;\mu,\omega)
\]
is independent of the choice of the Morse--Bott function $\mu$ within
a fixed conformal class. It vanishes if the codimension of $Y$ is odd,
so in this section we assume that $m$ is even. We define a sheaf of
differential forms whose global sections are the forms
$\mathcal A_\mu(X)$ on which $I_0$ is defined and give a local formula
for $I_0$ on a quasi-isomorphic subcomplex in terms of a residue map.
\subsection{The de Rham complex}\label{ss-3.1}
Let $Y\subset X$ as above, $U=X\smallsetminus Y$ and denote
$j\colon U\to X$ the inclusion map. We write
$\mathcal A_Z=\oplus_j\mathcal A_Z^j$ for the complex of sheaves of
differential forms on a manifold $Z$ with de Rham differential.  Let
$\mu$ be a nonnegative Morse--Bott function vanishing on $Y$, We
identify $\mathcal A_X$ as the subcomplex of $j_*\mathcal A_{U}$
consisting of forms on $U$ that extends to $X$.  We set
\begin{align*}
  \mathcal A_{X,\mu}&=\cup_{N\geq0}\mathcal A_{X,\mu,N},
  \\
  \mathcal A_{X,\mu,N}&=\{\omega\in j_*\mathcal
                        A_{X\smallsetminus Y}: \mu^N\omega\in\mathcal A_X\}.
\end{align*}
\begin{lemma}\label{l-01}
  \
\item{(i)} The de Rham differential maps $\mathcal A_{X,\mu,N}$ to
  $\mathcal A_{X,\mu,N+1}$.  In particular, $\mathcal A_{X,\mu}$ is a
  subcomplex of $j_*\mathcal A _{X\smallsetminus Y}$.
\item{(ii)} $\mathcal A_{X,\mu,N}$ depends only on the conformal class
  of $\mu$.
\end{lemma}

\begin{proof}
  (i) Suppose $\mu^N\omega$ extends smoothly to $X$ then
  $d(\mu^{N+1}\omega)$ is also smooth and therefore also
  \[
    \mu^{N+1}d\omega=d(\mu^{N+1}\omega)-(N+1)d\mu\wedge\mu^N\omega.
  \]
  The support condition is preserved by the differential. (ii) It is
  clear that $\mathcal A_{X,\mu}=\mathcal A_{X,f\mu}$ for any everywhere
  positive function $f$.
\end{proof}
In the notation of the preceding sections,
$\mathcal A_\mu(X)=\Gamma(X,\mathcal A_{X,\mu})$ is the differential
graded algebra of global sections.

We may then view $I_0$ as a map on compactly supported sections of
$\mathcal A^n_{X,\mu}$:
\[
  I_0\colon\Gamma_c(X,\mathcal A^n_{X,\mu})\to \mathbb C.
\]

\begin{prop}\label{pr-02}
  Let $\psi\in\Gamma_c(X,\mathcal A^{n-1}_{X,\mu})$. Then
  \begin{enumerate}
  \item[(i)] $I_0(\mu,d\psi)=0$.
  \item[(ii)] $I_{\mathrm{finite}}(\mu,d\psi)=I_0(\mu,\psi\wedge\theta)$,
    where $\theta=\frac{d\mu}{2\mu}$.
  \end{enumerate}
\end{prop}

\begin{proof} By Stokes's theorem we have
  \[
    \int_X\mu^{\frac s2}d\psi=-\frac s2\int_X\mu^{\frac
      s2-1}d\mu\wedge\psi,
  \]
  provided $\mathrm{Re}\,s$ is sufficiently large. Since
  $d\mu/\mu\wedge\psi$ belongs to $\mathcal A_{X,\mu}$, we get the
  identity of meromorphic functions
  \begin{equation}\label{e-desire}
    \zeta(s;\mu,d\psi) =
    - s\,\zeta\left(s;\mu,\frac{d\mu}{2\mu}\wedge\psi\right).
  \end{equation}
  Since the zeta function on the right has only simple poles, the
  left-hand side is regular at zero.  Since $I_0(\mu,d\psi)=0$, the finite
  part is just the value of the zeta function at $s=0$.  By
  \eqref{e-desire},
  \[
    \zeta(0;\mu,d\psi)=-\lim_{s=0}s\,\zeta(s;\mu,\theta\wedge\psi)
    =-\mathrm{res}_{s=0} \zeta(s;\mu,\theta\wedge\psi),
  \]
  which is $-I_0(\mu,\theta\wedge\psi)$ by definition.
\end{proof}

\subsection{The homomorphism to de Rham currents}
Let $\mathcal D^p(X)=\Gamma_c(X,\mathcal A^p_X)$ be the space of compactly supported differential
$p$-forms with the usual Fr\'echet topology. Recall that
the space $\mathcal D'(X)^p$ of de Rham currents of degree $p$ is the space
of continuous linear forms on $\mathcal D^{n-p}(X)$. The complex of currents
is the direct sum $\mathcal D'(X)=\oplus_{p=0}^n\mathcal D'(X)^p$ with differential
$d\colon\mathcal D'(X)^p\to \mathcal D'(X)^{p+1}$ defined by
\[
d \kappa (\varphi)=(-1)^{p+1} \kappa(d\varphi),
\quad \kappa\in \mathcal D'(X)^p,
\quad \varphi\in\mathcal D^{n-p}(X).
\]
For any smooth differential $p$-form $\omega$, the map
$\varphi\mapsto\int_X\omega\wedge\varphi$ defines a current of degree
$p$ and this defines an injective morphism of complexes
$\mathcal D(X)\hookrightarrow \mathcal D'(X)$. A current is said to be
supported on a closed subset $Y$ if it vanishes on all forms with
support in its complement.
\begin{theorem}\label{t-03}
  Let $\omega\in\Gamma(X,\mathcal A^p_{X,\mu})$. Then
  \[
    I_0(\mu,\omega\wedge\mbox{---})\colon\varphi\mapsto
    I_0(\mu,\omega\wedge\varphi),\quad \varphi\in\mathcal D^{n-p}(X),
  \]
  is a de Rham current supported on $Y$. The map
  $\omega\mapsto I_0(\mu,\omega\wedge\mbox{---})$ is a morphism of
  complexes $\Gamma(X,\mathcal A_{X,\mu})\to\mathcal D'(X)$.
\end{theorem}
\begin{proof} Suppose $\varphi$ has support in a sufficiently small
  neighbourhood of a point of $Y$. Then a local calculation with
  Morse--Bott coordinates shows that $I_0(\mu,\omega\wedge\varphi)$ is a
  finite sum of terms of the form $\int_Y\alpha\wedge D(\varphi)|_Y$
  for some differential operators $D$. By a partition of unity
  argument, the same holds for general $\varphi$. This is certainly a
  well-defined de Rham current. The fact that the map is a morphism of
  complexes follows from Prop.~\ref{pr-02}. Indeed let
  $\omega\in\Gamma(X,\mathcal A^p_{X,\mu})$,
  $\kappa_{\omega}=I_0(\mu,\omega\wedge\mbox{---})$ and
  $\varphi\in\mathcal D^{n-p-1}(X)$. Then
  \begin{align*}
    d\kappa_\omega(\varphi)&=(-1)^{p+1}I_0(\mu,\omega\wedge d\varphi)\\
                           &=-I_0(\mu,d(\omega\wedge\varphi))+I_0(\mu,d\omega\wedge\varphi)\\
                           &=0+\kappa_{d\omega}(\varphi).
  \end{align*}
\end{proof}

\subsection{The subcomplex of differential forms with tame
  singularities}
Let $m=2r$ be the even codimension of $Y$. We introduce a subcomplex
of the de Rham complex $\mathcal A_{X,\mu}$ which is quasi-isomorphic
to it. It is analogous to the complex of logarithmic forms.
\begin{definition} A differential form $\omega\in\mathcal A_{X,\mu}$
  has {\em tame singularities} if
  \begin{enumerate}
  \item[(i)] $\mu^{r}\omega\in\mathcal A_X$,
  \item[(ii)] $\mu^{r-1}d\mu\wedge\omega\in\mathcal A_X$,
  \end{enumerate}
  for the half-codimension $r$. We denote by
  $\mathcal A^{\mathrm{tame}}_{X,\mu}$ the sheaf of differential forms
  with tame singularities.
\end{definition}
In fact $\mathcal A^{\mathrm{tame}}_{X,\mu}$ is a subcomplex, as we
now show.  More generally we prove that it is part of a filtration of
the complex $\mathcal A_{X,\mu}$:
\[
  \cdots\subset F_p\mathcal A_{X,\mu}\subset F_{p+1}\mathcal
  A_{X,\mu}\subset\cdots\subset A_{X,\mu} =\cup_{p\in\mathbb
    Z}F_pA_{X,\mu}
\]
by subspaces
\[
  F_p\mathcal A_{X,\mu} = \{ \omega\in\mathcal A_{X,\mu} :
  \mu^{p}\omega, \mu^{p-1}d\mu\wedge\omega\in\mathcal A \}.
\]
\begin{lemma}\label{l-07}
  Each $F_p\mathcal A_{X,\mu}$, in particular
  $\mathcal A_{X,\mu}^{\mathrm{tame}}=F_r\mathcal A_{X,\mu}$, is a
  subcomplex.
\end{lemma}
\begin{proof}
  Suppose $\omega\in F_p\mathcal A_{X,\mu}$. Then
  \begin{align*}
    \mu^pd\omega
    &=d(\mu^p\omega)-p\mu^{p-1}d\mu\wedge\omega,
    \\
    \mu^{p-1}d\mu\wedge d\omega
    &=-d(\mu^{p-1}d\mu\wedge\omega).
  \end{align*}
  The right-hand sides are regular on $X$ by assumption.  Thus
  $F_p\mathcal A_{X,\mu}$ is a subcomplex.
\end{proof}
\begin{prop} Let $Y\subset X$ have codimension $m=2r$. The inclusion
  $\mathcal A_{X,\mu}^{\mathrm{tame}}=F_r\mathcal
  A_{X,\mu}\hookrightarrow \mathcal A_{X,\mu}$ is a quasi-isomorphism.
\end{prop}
This is a local statement, so it is sufficient to prove it on a small
ball in $\mathbb R^n$ with $\mu=\sum_{i=1}^mx_i^2$. In this case it
follows from the calculation of the cohomology done in the Appendix,
see Corollary \ref{c-02}.

\subsection{The residue map}
Let $i\colon Y\hookrightarrow X$ be the inclusion map and denote by
$\mathrm{or}_Y$ the orientation bundle of $Y$. We define a residue map
$R\colon \mathcal A_{X,\mu}^{\mathrm{tame}} \to
i_*(\mathrm{or}_Y\otimes\mathcal A_{Y}[-m])$ such that
$I_0(\mu,\omega)=\int_YR(\omega)$ for any top differential form
$\omega\in\Gamma_c(X,\mathcal A_{X,\mu}^{n,\mathrm{tame}})$ with tame
singularities and relatively compact support.  We denote by
$C^\infty_X$ the sheaf of smooth functions on $X$.
\begin{theorem}\label{t-04} Let $m=2r$ be the codimension of $Y$.
  There is a unique morphism of graded $C^\infty_X$-modules
  \[
    R\colon \mathcal A_{X,\mu}^{\mathrm{tame}}\to
    i_*(\mathrm{or}_Y\otimes\mathcal A_{Y}[-m])
  \]
  such that for any $p$-form
  $\omega\in\Gamma(X,\mathcal A_{X,\mu}^{\mathrm{tame}})$ with tame
  singularities and smooth compactly supported $(n-p)$-form
  $\varphi\in\Gamma(X,\mathcal A_X)$,
  \begin{equation}\label{e-31415926539}
    I_0(\mu,\omega\wedge\varphi)=
    \int_YR(\omega)\wedge\varphi.
  \end{equation}
  Moreover $R$ is a morphism of complexes of sheaves.
\end{theorem}
\dontprint{
  \begin{remark} The assumption of orientability of $Y$ is not
    necessary.  The construction of $R$ is local and depends on a
    local choice of orientation.  If we take the opposite orientation,
    $R$ is replaced by $-R$.  Thus $R$ defines a map to
    $i_*(\mathrm{Or}_Y\otimes \mathcal A_Y$, where $\mathrm{Or}_Y$ is
    the orientation bundle of $Y$, and the integration over $Y$ is
    well-defined even if $Y$ is not orientable.
  \end{remark}
} The proof of Theorem \ref{t-04} occupies the rest of this section.

We first discuss uniqueness. First of all for any global section
$\omega\in\Gamma_c(X,\mathcal A_{X,\mu}^{\mathrm{tame}})$, $R(\omega)$
is uniquely determined by \eqref{e-31415926539} since a de Rham
current is represented by at most one smooth form. It remains to show
that $R$ is uniquely determined by its action on global sections. This
follows from the fact that it is linear over the algebra of
functions: let $\omega$ be a section on an open set $U\subset X$. Then
for any $f\in C^\infty(U)$ with compact support, $f\omega$ extends (by
zero) to $X$ and since $R$ is a map of sheaves we obtain that
$I_0(\mu,f\omega\wedge\varphi)=\int_YR(f\omega)\wedge\varphi$ for all
$\varphi$ with support in $U$. We may now choose $f$ to be 1 on the
support of $\varphi$, so that $f\varphi=\varphi$. By the
$C^\infty_X$-linearity of $R$, it follows that \eqref{e-31415926539}
holds for sections $\omega$ on any open subset $U$ and $\varphi$ with
compact support in $U$. Therefore $R$ is uniquely defined as a map of
sheaves. The uniqueness also implies that $R$ is a morphism of
complexes: by Prop.~\ref{pr-02},
\begin{align*}
  I_0(\mu,d\omega\wedge\varphi)&=-(-1)^{p}I_0(\mu,\omega\wedge d\varphi)\\
                           &=-(-1)^{p}\int_Y R(\omega)\wedge d\varphi\\
                           &=(-1)^{m}\int_Y d R(\omega)\wedge \varphi,
\end{align*}
and therefore $R(d\omega)=(-1)^md R(\omega)$ ($(-1)^md$ is the
differential of $i_*\mathcal A_Y[-m]$).

To prove existence, we claim that we may assume that $X$ is a small
ball in $\mathbb R^{n}$ with the standard orientation and that
$\mu=x_1^2+\dots+x_m^2$. To reduce the general case to this local
statement, notice that we may an open cover $(U_i)$, such that on each
$U_i$ we have Morse--Bott coordinates. Assuming the local statement,
we obtain, for each global section $\omega$, forms $R_i(\omega)$
defined on $U_i$ such that
\[
  I_0(\mu,\omega\wedge\varphi)=\int_Y R_i(\omega)\wedge\varphi
\]
for all forms $\varphi$ with support on $U_i$. By uniqueness, these
forms $R_i(\omega)$ must agree on intersections and are thus
restrictions of a unique form $R(\omega)$ on $Y$.

From now on, we thus assume that $X$ is small ball around the origin
of $\mathbb R^n$ and that $\mu=\sum_{i=1}^mx_i^2$, so that $Y$ is
given by the equations $x_1=\cdots=x_m=0$.

\begin{lemma}\label{l-Varia}
  Let $\omega$ be a differential form with tame singularities. Then
  \[
    \omega=\frac1{(x_1^2+\cdots+x_m^2)^{\frac
        m2}}(dx_1\wedge\cdots\wedge
    dx_m\wedge\alpha+\sum_{i=1}^mx_i\alpha_i)
  \]
  for some smooth forms $\alpha,\alpha_i$.
\end{lemma}
\begin{proof}
  By the first condition for tame singularities,
  $\omega=\mu^{-\frac m2}\psi$ with smooth $\psi$. Let
  $dx^I=dx_{i_1}\wedge\cdots\wedge dx_{i_k}$ for $I=\{i_1<\dots<i_k\}$
  and set $|I|=k$.  $\psi=\sum_Idx^I\wedge\psi_I$ for some forms
  $\psi_I$ not involving $dx_1,\dots,dx_m$. The claim is that the
  second condition implies that $\psi_I$ vanishes on $Y$ for $|I|<m$
  and thus contributes to $\sum_{i=1}^mx_i\alpha_i$. To prove this
  claim, notice that the second condition may be written as
  \[
    \sum_{i=1}^mx_idx_i\wedge\psi\equiv 0\mod x_1^2+\cdots+x_m^2,
  \]
  implying that $\sum_{i=1}^mx_idx_i\wedge\psi|_{Y=0}=0$. This
  condition has the form
  \[
    \sum_{i\in I}\pm x_i\psi_{I\smallsetminus\{i\}}|_Y=0.
  \]
  Thus $\psi_J|_Y$ vanishes for all $J$ of the form
  $I\smallsetminus\{i\}$, namely such that $|J|<m$.
\end{proof}
We can now compute the residue of the zeta function in spherical
coordinates.  Let $x=(x',x'')$ with $x'$ the first $m$
coordinates. Write $x'=ry$ with $y\in S^{m-1}$ on the unit sphere with
volume form $d\Omega(y)$.  Then, in the notation of Lemma
\ref{l-Varia},
\[
  \zeta(s;\mu,\omega\wedge\varphi)=\int_{\mathbb{R}_{\geq0}\times
    S^{m-1}\times \mathbb R^{n-m}}r^{s-1}dr\wedge
  d\Omega(y)\wedge(\alpha\wedge\varphi+O(r)).
\]
This integral is a holomorphic function of $s$ in the right half-plane
and has a simple pole at $s=0$. Its residue can be computed as in the
proof of Theorem \ref{t-01} by first integrating over $r$. To do this
calculation we need to choose an orientation of $Y$ (a trivialization
of $\mathrm{or}_Y$), which we take to be defined by
$dx_{m+1}\wedge\cdots\wedge d x_n$ and the compatible orientation
$dx_1\wedge\cdots\wedge dx_m$ of the fibres.
\begin{align*}
  I_0(\mu,\omega\wedge\varphi) 
  &=
    \int_{S^{m-1}}d\Omega \int_{\mathbb R^{n-m}}  
    \alpha\wedge\varphi
  \\
  &= 
    \frac{2\pi^{\frac m2}}{(m/2-1)!}
    \int_Y\alpha\wedge\varphi.
\end{align*}
Thus the claim of the Theorem holds with
$R(\omega)=\frac{2\pi^{m/2}}{(m/2-1)!}\alpha|_Y$. If we change the
trivialization of $\mathrm{or}_Y$, the orientation of the fibres, and
thus $R$, change sign, and we get a well defined form on $Y$ twisted
by the orientation bundle. It is clear from the definition in Lemma
\ref{l-Varia} that $R$ is $C^\infty$-linear. The proof of Theorem
\ref{t-04} is complete.

\subsection{An explicit formula}\label{ss-3.5}
The proof of Lemma \ref{l-Varia} gives an explicit formula for
$R(\omega)$ in terms of Morse--Bott coordinates. We may formulate it
more invariantly as follows. The Hessian of the Morse--Bott function
$\mu$ defines a euclidean metric on the normal bundle. Each local
trivialization of $\mathrm{or}_Y$ defines an orientation of the normal
bundle. These two data define a volume form $V(\mu)$ on the normal
bundle, which we may view as a section of
$\mathrm{or}_Y\otimes\wedge^mT^*X|_Y$ vanishing on vectors tangent to
$Y$.

\begin{theorem}\label{t-05} Let $m=2r$ be the even codimension of $Y$.
  Let $\omega\in\Gamma(X,\mathcal A_{X,\mu}^{\mathrm{tame}})$. Then
  \[
    R(\omega)=\frac{2\pi^r}{(r-1)!}\frac{(\mu^{r}\omega)|_Y}{V(\mu)}.
  \]
  More properly, the restriction of $\mu^r\omega$ to $Y$ is of the
  form $V(\mu)\wedge \alpha$ and
  $R(\omega)=\frac{2\pi^r}{(r-1)!}\alpha$.
\end{theorem}

\subsection{Dependence on the conformal class of the Morse--Bott
  function}\footnote{This approach was suggested to us by Tomer
  Schlank}\label{ss-dependence}
To compare the complexes $\mathcal A_{X,\mu}$ for different $\mu$ we
notice that for any two such $\mu_0,\mu_1$, the function
$\mu\colon (x,t)\mapsto t\mu_1(x)+(1-t)\mu_0(x)$ is a nonnegative
Morse--Bott function on $X\times I$ vanishing on
$Y\times I\subset X\times I$, where $I=[0,1]$ and restricting to
$\mu_j$ at the endpoints.  Let $p\colon X\times I\to X$ be the
projection to the first factor.  We then have maps
\begin{equation}\label{e-face}
  \mathcal A_{X,\mu_0}\leftarrow p_*\mathcal A_{X\times
    I,\mu}\to\mathcal A_{X,\mu_1}.
\end{equation}

\begin{prop}\label{p-UrsuleMirouet}
  These maps are quasi-isomorphisms of complexes of sheaves.
\end{prop}

It follows that we have a canonical isomorphism between the cohomology
sheaves for $\mu_0$ and $\mu_1$.

To obtain a more precise information, in particular to show that the
composition of isomorphisms is again an isomorphism of this form, we
prove a slightly stronger version of this proposition: we denote by
$\Delta_p=\{t\in\mathbb R^{p+1}_{\geq0}\colon t_0+\dots+t_p=1\}$ the
geometric $p$-simplex.

Let $\mathit{MB}$ be the convex cone of nonnegative Morse--Bott
functions $Y$ $S(\mathit{MB})$ the category of simplices of the affine
singular set of $\mathit{MB}$. Its objects are affine $p$-simplices in
$\mathit{MB}$, i.e. affine maps from the geometric $p$-simplex
$\Delta_p=\{t\in\mathbb R^p_{\geq0}\colon \sum t_i=0\}$ to
$\mathit{MB}$, and the morphisms are compositions of face and
degeneracy maps.

\begin{theorem}\label{t-Ania} There is a functor
  \[
    F\colon S(\mathit{MB})\to \mathrm{ShDGA}(X)
  \]
  to the category of sheaves of differential graded algebras, such
  that on vertices
  \[
    F(\mu)=\mathcal A_{X,\mu},
  \]
  and sending all morphisms to quasi-isomorphisms.
\end{theorem}

Since $\mathit{MB}$ is contractible, it follows that all
$\mathcal A_{X,\mu}$ are quasi-isomorphic to the homotopy colimit
$\mathrm{hocolim}\,F$, and that the quasi-isomorphism is unique up to
a contractible space of choices.

To prove this theorem we begin by defining the functor.  Let
$\mu_0,\dots,\mu_p$ be nonnegative Morse--Bott functions vanishing on
$Y$. They are vertices of a $p$-simplex $\Delta(\mu_0,\dots,\mu_p)$,
an object of $S(\mathit{MB}).$ We set
\[
  F(\Delta(\mu_0,\dots,\mu_p))= p_*\mathcal A_{X\times \Delta_p,\sum
    t_i\mu_i}.
\]
Here $p\colon X\times\Delta_p\to X$ is the projection onto the first
factor and we view the convex linear combination
$\sum_{i=0}^p t_i\mu_i$ as a nonnegative Morse--Bott function on
$X\times\Delta_p$ vanishing on $Y\times \Delta_p$.  The face and
degeneracy maps are mapped to the pull-backs of corresponding face and
degeneracy maps on the $\Delta_p$.  The first example is
\eqref{e-face}.

We turn to the proof of Theorem \ref{t-Ania}.  Since the claim is a
local statement we may assume that $X$ is a small ball centered at the
origin in $\mathbb R^n$. The non-trivial case is when the origin is in
$Y$.  To prove the proposition in this case we need a slight
generalization of the Morse--Bott lemma.
\begin{lemma}\label{l-LioubovAndreevna}
  Let $\mu_0,\dots,\mu_p$ be nonnegative Morse--Bott functions on an
  open ball $B\subset \mathbb R^n$ centered at the origin and
  vanishing on the same smooth submanifold $0\in Y\subset B$ of
  codimension $m$. For $t\in\Delta_p$ set
  \[
    \mu_t=\sum_{i=0}^p t_i\mu_i.
  \] Then there are smooth functions $z_1,\dots, z_m$ on
  $B'\times \Delta_p$ for some possibly smaller ball $B'\subset B$,
  such that
  \[
    \mu_t=z_1^2+\cdots+z_m^2\; \text{on $B'\times \Delta_p$}.
  \]
\end{lemma}

\begin{proof}
  By the Morse--Bott lemma there exist functions
  $x^{(i)}_1,\dots,x^{(i)}_m$ vanishing on $Y$ and defined on a
  possibly smaller ball $B'\subset B$ and with linear independent
  differentials on $Y$, such that
  $\mu_i=(x^{(i)}_1)^2+\cdots+(x^{(i)}_m)^2$, for $i=0,\dots,p$. Let
  $x_1,\dots,x_m$ be generate the ideal of functions vanishing on $Y$,
  for instance $x_i=x_i^{(0)}$.  Then we can write
  $x^{(i)}_j=\sum_{k=1}^m x_ka^{(i)}_{kj}$, for some smooth functions
  $a^{(i)}_{kj}$ forming, for each $i$, a non-degenerate matrix. After
  possibly rotating $x^{(i)}_j$ by a linear orthogonal transformation,
  we may assume that $x^{(i)}_j=x_jg^{(i)}_j$ for some smooth
  functions $g^{(i)}_j$ with $g^{(i)}_j|_Y=1$. Then
  \[
    \mu_t=\sum_{j=1}^m x_j^2\sum_{i=0}^p t_i(g^{(i)}_j)^2.
  \]
  Since $\sum_{i=0}^p t_i(g^{(i)}_j)^2$ is close to 1 in the vicinity
  of $Y$, it is positive and we can define, again after making $B'$
  smaller, new functions
  \[
    z_j=x_j\sqrt{\sum_{i=0}^p t_i(g^{(i)}_j)^2}.
  \]
  such that $\mu_t=\sum_{j=1}^mz_j^2$.
\end{proof}

\noindent{\it Proof of Theorem \ref{t-Ania}}.  As we saw, it is
sufficient to assume that $X$ is a small ball centered at the origin
in $\mathbb R^n$. To prove that morphisms are mapped to
quasi-isomorphisms it is sufficient to prove that face maps and
degeneracy maps are mapped to quasi-isomorphisms. These maps involve
simplices with a fixed set of vertices, say $\mu_0,\dots,\mu_p$. By
Lemma \ref{l-LioubovAndreevna} we may assume that
$\mu_i=\bar\mu=x_1^2+\cdots+x_m^2$ for all $i$. Thus
\[
  \sum_{i=0}^pt_i\mu_i=\bar\mu,
\]
is a constant function of $t\in\Delta_p$. Thus the algebras
$F(\Delta(\mu_{i_0},\dots,\mu_{i_k})$ are all equal to
$p_*\mathcal A_{X\times\Delta_k,\bar\mu}$. The maps
$\mathcal A_{X,\bar\mu}\to p_*\mathcal A_{X\times\Delta_p,\bar\mu}$
sending a form to its pull back by $p$ are quasi-isomorphisms
commuting with face and degeneracy maps.  The maps induced by face and
degeneracy maps in cohomology are thus the identity maps in
$\mathcal H(\mathcal A_{X,\bar\mu}$.  \hfill$\square$

We conclude this section by stating an elementary consequence.

\begin{cor} \label{c-Lopachkine} For any two nonnegative Morse--Bott
  functions $\mu_0,\mu_1$ we have a canonical isomorphism of the
  cohomology sheaves
  \[
    I(\mu_0,\mu_1)\colon\mathcal H(\mathcal A_{X,\mu_0})\to\mathcal
    H(\mathcal A_{X,\mu_1}).
  \]
  For any three $\mu_0,\mu_1,\mu_2$ we have
  \[
    I(\mu_1,\mu_2)\circ I(\mu_0,\mu_1)=I(\mu_0,\mu_2)
  \]
\end{cor}

The first statement follows from Proposition
\ref{p-UrsuleMirouet}. The second statement follows from the case
$p=2$ of Theorem \ref{t-Ania}. The sheaves
$\mathcal A_i=\mathcal A_{X,\mu_i}$, $i=0,1,2$ are related by a
commutative diagram of quasi-isomorphisms:
\[
  \begin{tikzpicture}[baseline=0,xscale=.6,yscale=.6]
    \draw (0,-0.1) node {$\mathcal
      A_0$}; \draw [<-] (.5*0.5,.5*0.866) -- (3*.5,3*.866); \draw
    (3.5*0.5,3.5*0.866-.1) node {$\mathcal
      A_{01}$}; \draw [->] (4*0.5,4*0.866) -- (6.5*.5,6.5*.866); \draw
    (7*.5,7*.866-.1) node {$\mathcal
      A_1$}; \draw (7,0-.1) node {$\mathcal
      A_2$}; \draw [<-] (7-.5*0.5,.5*0.866) -- (7-3*.5,3*.866); \draw
    (7-3.5*0.5,3.5*0.866-.1) node {$\mathcal
      A_{12}$}; \draw [->] (7-4*0.5,4*0.866) -- (7-6.5*.5,6.5*.866);
    \draw (7-7*.5,7*.866-.1) node {$\mathcal
      A_1$}; \draw [<-] (.5,0) -- (3,0); \draw (3.5,0-.1) node
    {$\mathcal
      A_{02}$}; \draw [->] (4,0) -- (6.5,0); \draw (3.5,7*0.289-.1)
    node {$\mathcal
      A_{012}$}; \draw [->] (3.5,7*0.289-.5) -- (3.5,0.5); \draw [->]
    (3.5-.5*.866,7*0.289+.25) -- (3.5*0.5+.5*.866,3.5*0.866-.25);
    \draw [->] (3.5+.5*.866,7*0.289+.25) --
    (7-3.5*0.5-.5*.866,3.5*0.866-.25);
  \end{tikzpicture}
\]
where
$A_{i_0,\dots,i_k}=p_*\mathcal A_{X\times\Delta_k,\sum t_s\mu_{i_s}}$.
\section{Complex hypersurfaces}\label{s-4}
\subsection{The canonical conformal class of Morse--Bott functions}
Suppose $D\subset X$ is a smooth divisor (complex hypersurface) in a
$d$-dimensional complex manifold. Thus we have $n=2d$ and $m=2$. The
complex structure defines a canonical conformal class of nonnegative
Morse--Bott functions $\mu$ vanishing on $D$: locally on an open set
$U\subset X$, $D$ is defined by $f=0$ for some holomorphic
function $f$ on $U$ such that $df|_{D\cap U}\neq0$. We call such a function
a local equation for $D$.  We then require
$\mu/|f|^2$ to extend to a positive smooth function. This condition is
independent of $f$ as the ratio of any two $f$'s is a nowhere
vanishing function.  Any two nonnegative Morse--Bott functions with
this property differ by multiplication by a positive function and thus
define a conformal class. The differential forms we consider can then
be defined as those locally of the form $\omega/|f|^{2N}$ with
$\omega$ smooth and $f$ a local equation for $D$.

The de Rham complex $\mathcal A_{X,\mu}$ is quasi-isomorphic to the
subcomplex of differential forms with tame singularities.  Here is a
description of these forms, which in this context could be called
bilogarithmic.
\begin{prop}\label{p-Galleria di base 21 ottobre 2016!}
  Let $\omega\in\Gamma(X,\mathcal A_{X,\mu})$. Then $\omega$ has tame
  singularities if and only of for any local equation $f$
  of $D$,
  \[
    \omega=\frac {df}f\wedge\frac {d\bar f}{\bar f}\wedge\omega_{1,1}
    +\frac{df}f\wedge\omega_{1,0}+\frac {d\bar f}{\bar
      f}\wedge\omega_{0,1}+\omega_{0,0},
  \]
  for some smooth forms $\omega_{i,j}$.  The residue map is
  \[
    R(\omega)=-4\pi i\omega_{1,1}.
  \]
\end{prop}
\begin{proof}
  We may choose local complex coordinates $z_1,\dots,z_d$ so that
  $f=z_1$ and $\mu=|z_1|^2$. Let us write
  \[
    |z_1|^2\omega=dz_1\wedge d\bar
    z_1\wedge\alpha_{1,1}+dz_1\wedge\alpha_{1,0} +d\bar
    z_1\wedge\alpha_{0,1}+\alpha_{00},
  \]
  for some forms $\alpha_{i,j}$ not involving $dz_1$ or $d\bar z_1$.
  The first condition for tameness implies that $\alpha_{i,j}$ are
  smooth.  The second condition in real codimension 2 states that
  $ d\mu\wedge\omega $ is smooth. Since
  $d\mu=z_1d\bar z_1+\bar z_1dz_1$ this translates to the smoothness
  of
  \[
    \omega_{1,0}=\frac1{\bar z_1}\,\alpha_{1,0},\quad\omega_{0,1}=\frac1{
      z_1}\,\alpha_{0,1},\quad \omega_{0,0}=\frac1{z_1\bar
      z_1}\alpha_{0,0}.
  \]
  Comparing with the explicit formula of \ref{ss-3.5}, we see that
  $R(\omega)$ is proportional to $\alpha$. The Morse--Bott coordinates
  $x_1,x_2$ are given by $z_1=x_1+ix_2$ and
  $dz_1\wedge d\bar z_1/|z_1|^2=-2i\,dx_1\wedge
  dx_2/(x_1^2+x_2^2)$. Thus
  \[
    R(\omega)=2\pi (-2i)\omega_{1,1}=-4\pi i\omega_{1,1}.
  \]
\end{proof}
In \cite{FelderKazhdan2016}, inspired by calculations in perturbative
superstring theory, we considered the case where
$\omega=\alpha\wedge\bar\beta$ where $\beta$ is a holomorphic $d$-form
with simple pole on $Y$ and $\alpha$ is a smooth $(d,0)$-form on
$X\smallsetminus D$ so that locally $f^N\alpha$ extends to a compactly
supported smooth form on $X$ for some $N$.  There we defined a
Dolbeault residue $\mathrm{Res}_\partial$ defined on this class of
$(d,0)$-forms $\alpha$ and taking values in $\partial$-cohomology
classes of forms of type $(d-1,0)$ on $D$. The Dolbeault residue
vanishes on $\partial$-exact forms and coincides with the Poincar\'e
residue $\mathrm{Res}$ for forms with first order pole. Comparing with
Prop.~\ref{p-Galleria di base 21 ottobre 2016!}  we obtain
\[
  R(\alpha\wedge\bar\beta)=4\pi
  i(-1)^{d}\mathrm{Res}_\partial\alpha\wedge\overline{\mathrm{Res}\,\beta}.
\]
The dependence on the Morse--Bott function of the Hadamard finite part
is thus
\[
  I_{\mathrm{finite}}(\mu
  e^{2\varphi},\alpha\wedge\bar\beta)=I_{\mathrm{finite}}(\mu,\alpha\wedge\bar\beta)
  + (-1)^{d}4\pi
  i\int_\varphi\mathrm{Res}_\partial\alpha\wedge\overline{\mathrm{Res}\,\beta},
\]
in agreement with \cite{FelderKazhdan2016}.

\section{Normal crossing divisor}\label{ss-nc}
It is desirable to extend our results to the case where the
singularities of the differential forms are not smooth
submanifolds. We consider here the special case of a divisor $D$ with
normal crossings in a complex manifold. We first focus
on the case of two components $D=D_1\cup D_2$. Away from the intersection
the theory of Section \ref{s-4} applies, so it is sufficient to
consider a neihgbourhood of the intersection, which is locally given by
$z_1=0$, $z_2=0$, for some local coordinate functions $z_1$, $z_2$.
Let $\omega$ be a top degree form on $X\smallsetminus D$ and assume
that $|z_1z_2|^{2N}\omega$ extends to a smooth form on $X$ with
compact support.  The zeta function is
\[
  \zeta(s_1,s_2;\mu_1,\mu_2,\omega)=\int_X\mu_1^{\frac{s_1}2}\mu_2^{\frac{s_2}2}
  \omega.
\]
It depends on nonnegative Morse--Bott functions $\mu_1,\mu_2$
and vanishing on $D_1$ and $D_2$
respectively. As in the case of a smooth divisor we take $\mu_1$
and $\mu_2$ in the canonical conformal class defined  by the complex
structure. As a function of $s_1,s_2$ the zeta function is
holomorphic for $\mathrm{Re}(s_i)$ large enough and extends to a
meromorphic function on $\mathbb C^2$ with at most simple poles on the
lines $s_i=2k$, $k\in\mathbb Z$.

We define the finite part of $\int_X\omega$ as the constant term of
the Laurent expansion of $\zeta$ at 0.
\begin{definition} The {\em finite part} of the divergent integral 
$\int_X\omega$ is
\[
  I_{\mathrm{finite}}(\mu_1,\mu_2,\omega)=\mathrm{res}_{s_1=0}\mathrm{res}_{s_2=0}
  \frac1{s_1s_2}\zeta(s_1,s_2;\mu_1,\mu_2,\omega).
\]
\end{definition}
\begin{remark}
  If $\omega$ is regular on one of the components, say $D_2$, then
  $\zeta$ is regular at $s_2=0$ and our definition of the finite part
  reduces to the one for smooth divisors.
\end{remark}
To describe the dependence on the Morse--Bott functions it is useful
to introduce coefficients of divergent terms:
\[
  I_{j,k}(\mu_1,\mu_2,\omega)
  =\mathrm{res}_{s_1=0}\mathrm{res}_{s_2=0}\frac1{s_1^{j}s_2^{k}}
  \zeta(s_1,s_2;\mu_1,\mu_2,\omega).
\]
We only care about $j,k=0$ or $1$. The Laurent expansion looks like
\[
  \zeta(s_1,s_2;\mu_1,\mu_2,\omega)=
  \frac{I_{0,0}}{s_1s_2}+\frac{I_{0,1}}{s_1}+
  \frac{I_{1,0}}{s_2}+I_{\mathrm{finite}}+\cdots
\]
(in the dots there are other divergent terms such as $s_1/s_2$). Note
that $I_{\mathrm{finite}}=I_{1,1}$.

\begin{prop}
  \
  \begin{enumerate}
  \item $I_{0,0}$ is independent of $\mu_1,\mu_2$, $I_{0,1}$ is
    independent of $\mu_1$ and $I_{1,0}$ is independent of $\mu_2$.
  \item Let $\varphi\in C^\infty(X)$. Then
    \begin{align*}
      I_{\mathrm{finite}}(e^{2\varphi}\mu_1,\mu_2,\omega)
      &=I_{\mathrm{finite}}(\mu_1,\mu_2,\omega)+I_{0,1}(\mu_2,\varphi\,\omega),\\
      I_{\mathrm{finite}}(\mu_1,e^{2\varphi}\mu_2,\omega)
      &=I_{\mathrm{finite}}(\mu_1,\mu_2,\omega)+I_{1,0}(\mu_1,\varphi\,\omega),\\
      I_{1,0}(e^{2\varphi}\mu_1,\omega)
      &=I_{1,0}(\mu_1,\omega)+I_{0,0}(\varphi\,\omega),\\
      I_{0,1}(e^{2\varphi}\mu_2,\omega)
      &=I_{0,1}(\mu_2,\omega)+I_{0,0}(\varphi\,\omega).
    \end{align*}
  \end{enumerate}
\end{prop}

\begin{proof}
  For $s_1,s_2$ with large real part
  \begin{align*}
    \zeta(s_1,s_2;e^{2\varphi}\mu_1,\mu_2,\omega)-
    \zeta(s_1,s_2;\mu_1,\mu_2,\omega)
    &=s_1\int_X\mu_1^{\frac{s_1}2}\mu_2^{\frac{s_2}2}\frac{e^{s_1\varphi}-1}{s_1}\omega
    \\
    &=s_1\int_X\mu_1^{\frac{s_1}2}\mu_2^{\frac{s_2}2}\varphi\,\omega+(\dots),
  \end{align*}
  The analytic continuation is regular on the line $s_1=0$ and
  $(\dots)$ vanishes there.  For $j\in\{0,1\}$ we get
  \begin{align*}
    I_{0,j}(e^{2\varphi}\mu_1,\mu_2,\omega)
    -I_{0,j}(\mu_1,\mu_2,\omega)
    &=0,\\
    I_{1,j}(e^{2\varphi}\mu_1,\mu_2,\omega)
    -I_{1,j}(\mu_1,\mu_2,\omega)
    &=I_{0,j}(\mu_1,\mu_2,\varphi\,\omega),
  \end{align*}
  and similarly for $\mu_2$.
\end{proof}
We thus obtain the formula
\begin{align*}
  I_{\mathrm{finite}}(e^{\varphi_1}\mu_1,e^{\varphi_2}\mu_2,\omega)
  &=I_{\mathrm{finite}}(\mu_1,\mu_2,\omega)+I_{1,0}(\mu_1,{\varphi_2}\omega)
  \\ &+I_{0,1}(\mu_2,{\varphi_1}\omega)+I_{0,0}({\varphi_1}{\varphi_2}\omega).
\end{align*}
It is easy to generalize this result to the case of an arbitrary
number of components.  Let $D=D_1\cup\cdots\cup D_m$ be a divisor in
$X$ with normal crossings and irreducible components $D_i$. Take a
nonnegative Morse--Bott function $\mu_i$ for $D_i$ in the canonical 
conformal class for each $i$. Let $\omega$ be
a top form on $X\smallsetminus D$ so that, at a generic point of each
component $D_i$, $\mu_i^{N}\omega_i$ is smooth for sufficiently large
$N$. We then have a zeta function
\[
  \zeta(s;(\mu_i)_{i=1}^m,\omega)=\int_X\mu_1^{\frac{s_1}2}\cdots\mu_m^{\frac{s_m}2}\omega
\]
which has an analytic continuation to a meromorphic function of
$s\in\mathbb C^m$ with at most simple poles on the hyperplanes
$s_i=2k$, $i=1,\dots,m$, $k\in\mathbb Z$. We let $[m]=\{1,\dots,m\}$
and for any $M\subset [m]$,
\[
  I_M((\mu_i)_{i=1}^m,\omega)=\mathrm{res}_{s_1=0}\cdots\mathrm{res}_{s_m=0}
  \left(\prod_{i\in
      M}\frac1{s_i}\right)\,\zeta(s;(\mu_i)_{i=1}^m,\omega).
\]
and define the finite part as
\[
  I_{\mathrm{finite}}((\mu_i)_{i=1}^m,\omega)=I_{[n]}((\mu_i)_{i=1}^m,\omega).
\]
\begin{theorem}\label{t-06}
  \
  \begin{enumerate}
  \item $I_M$ is independent of $\mu_j$, $j\not\in M$.
  \item Let $i\in M$ and $\varphi\in C^\infty(X)$. Then
    \[
      I_M(\dots, e^{2\varphi}\mu_i,\dots,\omega)=
      I_M(\dots,\mu_i,\dots,\omega) +
      I_{M\smallsetminus\{i\}}(\mu_1,\dots,\mu_m,\varphi\,\omega)
    \]
  \end{enumerate}
\end{theorem}
We write $I_M=I_M((\mu_i)_{i\in M},\omega)$ accordingly.
\begin{cor}
  \begin{align*}
    I_{\mathrm{finite}}((e^{2\varphi_i}\mu_i)^m_{i=1},\omega)
    &=
      \sum_{M\subset[m]}I_M((\mu_i)_{i\in M},
      {\textstyle{\prod_{i\not\in M}}}\varphi_i\,\omega)\\
    &=
      I_{\mathrm{finite}}((\mu_i)_{i=1}^m,\omega)
      +\sum_{M\subsetneqq[m]}I_M((\mu_i)_{i\in M},
      {\textstyle{\prod_{i\not\in M}}}\varphi_i\,\omega).
  \end{align*}
\end{cor}

\section{Codimension one, manifolds with boundary}\label{s-5}
In the case of codimension 1 there are two situations in which it
makes sense to ask about divergent integrals with an integrand singular
on $Y\subset X$: the case of a hypersurface $Y$ in a manifold $X$,
which is a special case of what we considered so far, and the case of a
manifold $X$ with boundary (or a boundary component) $Y$, to which we
show that our results extend.
\subsection{Real hypersurfaces}
In the case of a submanifold $Y\subset X$ of codimension $m=1$, the
ratio of any two nonnegative Morse--Bott functions vanishing on $Y$ is
an everywhere positive function, so there is only one conformal class
in this case. It is a special case of odd codimension, and thus
$I_0=0$.  By Theorem \ref{ti-1}, the finite part of a divergent
integral $\int_X\omega$ is the value of the zeta function at $s=0$, it
is independent of the choice of Morse--Bott function and depends only
on the cohomology class of $\omega\in\mathcal A_\mu(X)$.  This is a
generalization of the classical theory of principal values.
\subsection{Manifolds with boundary} Let $X$ be an $n$-dimensional
oriented manifold with boundary and $Y\subset X$ be a union of
connected components of the boundary. Let $\lambda$ be a nonnegative
smooth function on $X$ vanishing to first order on $Y$ (i.e., such that
$\lambda|_Y=0$ and $d\lambda|_Y\neq0$) and positive
everywhere else. Such a function is unique up to multiplication by an
everywhere positive function.  A differential form $\omega$ defined on
$X\smallsetminus Y$ is said to have a pole singularity of order
$M\in\mathbb Z_{\geq0}$ on the boundary if $\lambda^M\omega$ extends
smoothly%
\footnote{A smooth form on a manifold with boundary is by definition a
  differential form whose pull-back to any coordinate chart
  $U\subset\mathbb R_{\geq0}\times\mathbb R^{n-1}$ is locally the
  restriction of a form defined on $V\subset\mathbb R^{n}$ for some
  open set $V\supset U$.} %
to the boundary for some integer $M$. This condition does not depend
on the choice of $\lambda$.  Let $\mathcal A_{X,Y}$ be the sheaf of
differential forms on the interior of $X$ with polar singularities on
the boundary. We consider regularization of divergent integrals
\[
  \int_{\lambda\geq\epsilon}\omega, \qquad
  \omega\in\Gamma_c(X,\mathcal A^n_{X,Y})
\]
on global sections with compact support.  The corresponding zeta
function is the meromorphic continuation of
\[
  \tilde\zeta(s;\lambda,\omega)=\int_X\lambda^s\omega,\qquad
  \mathrm{Re}\,s\gg 0
\]
To compare with the previous sections, note that $\mu=\lambda^2$
vanishes to second order at the boundary and should be thought of as a
nonnegative Morse--Bott function. Then
$\tilde\zeta(s;\lambda,\omega)=\zeta(s;\mu,\omega)$ and
$\mathcal A_{X,Y}$ is $\mathcal A_{X,\mu}$; it is independent of
$\mu$.

The proofs of the following results are parallel to the ones in the
case of even codimension of the previous sections. The main difference
is that the poles of the zeta function lie on arithmetic progressions
with step 1 rather than 2. Therefore the zeta function has a pole at
zero in spite of the fact that the codimension is odd.

\begin{theorem}\label{t-b01} Let $M\in\mathbb Z_{\geq0}$ and $\lambda$
  be a nonnegative smooth function vanishing on $Y$ to first
  order. Let $\omega\in\Gamma_c(X,\mathcal A_{X,Y})$ have polar
  singularity of order $M$ at $Y\subset \partial X$. Then
  \begin{enumerate}
  \item[(i)] $\int_X\lambda^{s}\omega$ is holomorphic for
    $\mathrm{Re}\,s>M-1$ and has a meromorphic continuation
    $\tilde\zeta(s;\lambda,\omega)$ with at most simple poles on the
    arithmetic progression $s=M-1,M-2,\dots$.
  \item[(ii)] As $\epsilon\to 0$ we have an expansion
    \[
      \int_{\lambda\geq\epsilon}\omega
      =\sum_{k=1}^{M-1}I_{-k}(\lambda,\omega)\epsilon^{-k}+
      I_0(\omega)\log\frac1\epsilon
      +I_{\mathrm{finite}}(\lambda,\omega)+O(\epsilon).
    \]
  \item[(iii)] For $k=1,\dots, M-1$,
    \[
      I_{-k}(\lambda,\omega)
      =\frac1k\mathrm{res}_{s=k}\tilde\zeta(s;\lambda,\omega),
    \]
    and
    \[
      I_0(\omega) =\mathrm{res}_{s=0}\tilde\zeta(s;\lambda,\omega)
    \]
    is independent of $\lambda$. The finite part is
    \[
      I_{\mathrm{finite}}(\lambda,\omega) =\lim_{s\to
        0}\left(\tilde\zeta(s;\lambda,\omega)-\frac{I_0(\omega)}{s}\right).
    \]
  \item[(iv)] For any function $\varphi\in C^\infty(X)$,
    \[
      I_{\mathrm{finite}}(\lambda
      e^{\varphi},\omega)=I_{\mathrm{finite}}(\lambda,\omega) +
      I_0(\varphi\,\omega).
    \]
  \end{enumerate}
\end{theorem}
The analogue of the tame differential forms are (the real version) of
logarithmic forms. By definition, a logarithmic form in
$\mathcal A_{X,Y}$ is a form $\omega$ such that $\lambda\,\omega$ and
$d\lambda/\lambda\wedge\omega$ extend to smooth forms on $X$ for one
(and thus any)
choice of a nonnegative function $\lambda$ vanishing to first order
on $Y$. As in the complex case, and in the case of forms
with tame singularities, logarithmic forms form a subcomplex of
sheaves $\mathcal A_{X,Y}^{\log}$ which is quasi-isomorphic to
$\mathcal A_{X,Y}$. Given a choice of the function $\lambda$ vanishing
to first order on $Y\subset\partial X$, any logarithmic form can
locally be written as
\[
  \omega=\frac{d\lambda}{\lambda}\wedge\sigma+\tau,
\]
for some smooth forms $\sigma$, $\tau$. Moreover it is standard to
check that $\sigma|_Y$ is independent of the choice of the decomposition
and of the choice of $\lambda$. Thus the map $\omega\mapsto \sigma|_Y$ is
well-defined and is the real analogue of the Poincar\'e residue map.
\begin{definition} The {\em residue} is the map
  $R\colon\mathcal A_{X,Y}\to i_*\mathcal A_Y[-1]$ such that
  \[
    R\left(\frac{d\lambda}{\lambda}\wedge\sigma+\tau\right)=\sigma|_Y.
  \]
\end{definition}
\begin{theorem}\label{t-b04}
  \
  \begin{enumerate}
  \item[(i)] The residue map $R$ is a morphism of complexes of
    sheaves.
  \item[(ii)] For any logarithmic $p$-form
    $\omega\in\Gamma(X,\mathcal A_{X,Y}^{\log})$ smooth compactly
    supported $(n-p)$-form $\varphi\in\Gamma(X,\mathcal A_X)$,
    \[
      I_0(\omega\wedge\varphi)= \int_YR(\omega)\wedge\varphi.
    \]
  \end{enumerate}
\end{theorem}
The real analogue of a normal crossing divisor is the boundary of a
manifold with corners.  We leave it to the reader to extend the
results of Section \ref{ss-nc} to this case.
\appendix

\section{Cohomology: local calculation}
Let $X$ be an open ball in $\mathbb R^n$ centered at the origin and
$Y\subset X$ its intersection with the subspace $x_{1}=\cdots=x_m=0$.
Let $\mu=x_1^2+\cdots+x_m^2$. We compute the cohomology of the complex
$\mathcal A_\mu(X)$ of differential forms $\alpha$ on
$X\smallsetminus Y$ such that $\mu^N\alpha$ is smooth for some $N$.

\subsection{Cohomology of $\mathcal A_\mu(D^n)$}
We denote by $\bigwedge(t_1,\dots,t_k)$ the exterior algebra with
generators $t_1,\dots,t_k$.

\begin{prop}\label{p-01}
  Let $X=D^n$ be an open ball in $\mathbb R^n$ centered at the origin,
  $\mu=x_1^2+\cdots+x_m^2$, and $Y=\mu^{-1}(0)\cap X\subset X$.
  \begin{enumerate}
  \item[(i)] If $m$ is odd,
    \[
      H(A_\mu(X))\cong \bigwedge(\bar\alpha),\quad \mathrm{deg}(\bar
      \alpha)=1.
    \]
  \item[(ii)] If $m$ is even,
    \[
      H(A_\mu(X))\cong \bigwedge(\bar\alpha,\bar\beta),\quad
      \mathrm{deg}(\bar\alpha)=1,\quad\mathrm{deg}(\bar\beta)=m-1.
    \]
  \end{enumerate}
  Here $\bar\alpha$ is the class of $\alpha=d\mu/\mu$ and $\bar\beta$
  (in the even case) is the class of
  \[
    \beta=
    \sum_{i=1}^{m}(-1)^{i-1}\frac{dx_1\wedge\cdots\wedge\widehat{dx_i}\wedge\cdots\wedge
      dx_m}{(x_1^2+\cdots+x_m^2)^{m/2}}.
  \]
  It is the basic representative of a rotation invariant volume form
  on the $m-1$-dimensional real projective space.
\end{prop}
To prove this result we first notice that $\mathcal A_\mu(X)$ has a
subcomplex $\mathcal B(X)$ of differential forms vanishing to infinite
order at $Y$. By Borel's lemma, the quotient
$\mathcal A_\mu(X)/\mathcal B(X)$ is isomorphic to the complex of
differential forms that are formal power series in the normal
direction:
\[
  \mathcal C_\mu(X)=\mathcal A_\mu(X)/\mathcal B(X)=\mathcal
  A(Y)[[x_1,\dots,x_m]][\frac1\mu]\otimes\bigwedge(dx_1,\dots,dx_m).
\]
The Euler vector field $e=\sum_{i=1}^mx_i\partial_{x_i}$ acts on
$\mathcal C_\mu$ via the Lie derivative
$L_e=d\circ\iota_e+\iota_e\circ d$.
\begin{lemma}\label{l-001}
  The inclusion map
  $\mathrm{Ker}(L_e)\hookrightarrow \mathcal C_\mu(X)$ induces an
  isomorphism in cohomology.
\end{lemma}
\begin{proof}
  The complex $\mathcal C_\mu$ splits into a direct sum
  $\mathcal C_\mu(X)=\mathrm{Ker}(L_e)\oplus\mathrm{Im}(L_e)$ of
  subcomplexes invariant under $L_e$, such that $L_e$ is invertible on
  $\mathrm{Im}(L_e)$.
\end{proof}
\begin{lemma}\label{l-002}
  Any form $\omega\in\mathcal C_\mu(X)$ can uniquely be written as
  \[
    \omega=\frac{d\mu}\mu\wedge\sigma+\tau,
  \]
  where $\iota_e\sigma=0=\iota_e\tau$.
\end{lemma}

\begin{proof}
  We have
  \[
    \iota_ed\mu=\sum x_i\iota_{\partial_{x_i}}2\sum{x_idx_i}=2\mu.
  \]
  Thus, given $\omega$ we set $\sigma=\frac12\iota_e\omega$ and
  $\tau=\omega-\frac{d\mu}{\mu}\wedge\sigma$.  Then $\iota_e\sigma=0$
  (since $\iota_e^2=0$) and
  \[
    \iota_e\tau=\iota_e\omega-\frac12\iota_e\frac{d\mu}{\mu}\wedge\sigma=0.
  \]
  This proves existence. To check uniqueness, suppose
  $0=\frac{d\mu}{\mu}\wedge\sigma+\tau$ with
  $\sigma,\tau\in\mathrm{Ker}(\iota_e)$.  Then applying $\iota_E$ we
  get $0=\iota_E\frac{d\mu}\mu\wedge\sigma$ and thus $\sigma=0$, and
  therefore also $\tau=0$.
\end{proof}
Let
$\mathcal C_{\mu,\mathrm{basic}}(X)=\mathrm{Ker}(L_e)\cap
\mathrm{Ker}(\iota_e)$. It is the subcomplex of basic differential
forms for the action of the group of dilations in the normal
direction.  Let us denote by $C[-1]$ the $(-1)$-shift of a cochain
complex $C$: $C[-1]^i=C^{i-1}$ with differential $d_{C[-1]}=-d_C$.
\begin{lemma}\label{l-003}
  The map
  \[
    \mathcal C_{\mu,\mathrm{basic}}(X)[-1]\oplus\mathcal
    C_{\mu,\mathrm{basic}}(X)\to \mathrm{Ker}(L_e)
  \]
  sending $\sigma\oplus\tau$ to $d\mu/\mu\wedge\sigma+\tau$ is an
  isomorphism of complexes.
\end{lemma}
\begin{proof}
  If $\omega=d\mu/\mu\wedge\sigma+\tau$ with basic forms
  $\sigma,\tau$, then
  \[
    d\omega=d\mu/\mu\wedge (-d\sigma)+d\tau,
  \]
  thus the map is compatible with differentials. By the uniqueness
  part of Lemma \ref{l-002} it is injective. To prove surjectivity,
  suppose $\omega\in\mathrm{Ker}(L_e)$.  Then by Lemma \ref{l-002},
  $\omega=d\mu/\mu\wedge\sigma+\tau$ with
  $\sigma, \tau\in\mathrm{Ker}(\iota_e)$ and applying $L_e$ we see
  that
  \[
    0=\frac{d\mu}{\mu}\wedge L_e\sigma+L_e\tau,
  \]
  and, since $\iota_e$ commutes with $L_e$, also
  $L_e\sigma,L_e\tau\in\mathrm{Ker}(\iota_e)$.  Again by the
  uniqueness part of Lemma \ref{l-002}, it follows that
  $L_e\tau,L_e\sigma$ both vanish.
\end{proof}

The rotation group $\mathrm{SO}(m)$ acts on
$\mathcal C_{\mu,\mathrm{basic}}(X)$ and by averaging we can replace
this complex by the quasi-isomorphic subcomplex of
invariants. Combining Lemma \ref{l-003} with Lemma \ref{l-001}, we
get:
\begin{cor}\label{c-01}
  The map
  \[
    \mathcal
    C_{\mu,\mathrm{basic}}(X)^{\mathrm{SO}(m)}[-1]\oplus\mathcal
    C_{\mu,\mathrm{basic}}(X)^{\mathrm{SO}(m)}\to\mathcal C_\mu(X)
  \]
  sending $\sigma\oplus\tau$ to $d\mu/\mu\wedge\sigma+\tau$ is a
  quasi-isomorphism.
\end{cor}

\begin{lemma}\label{l-004} Let $\beta$ be the closed differential form
  defined in Prop.~\ref{p-01}.  As a module over $\mathcal A(Y)$,
  \[
    \mathcal C_{\mu,\mathrm{basic}}(X)^{\mathrm{SO}(m)}=\begin{cases}
      \mathcal A(Y)1,&\text{if $m$ is odd,}\\
      \mathcal A(Y)1\oplus\mathcal A(Y)\beta,&\text{if $m$ is even.}
    \end{cases}
  \]
\end{lemma}
\begin{proof} The complex $\mathcal C_{\mu,\mathrm{basic}}(X)$
  consists of basic homogeneous differential forms in
  $\mathcal C_\mu(X)$.  They are thus homogeneous rational
  differential forms in the normal variables $x_1,\dots,x_m$ with
  coefficients in $\mathcal A(Y)$ and with powers of $\mu$ as
  denominators. They can be viewed as differential forms on
  $(\mathbb R^m\smallsetminus\{0\})\times Y$ that are basic for the
  action of the group $\mathbb R_{>0}$ of dilations. They thus define
  differential forms on the quotient $S^{m-1}\times Y$. The only
  $\mathrm{SO}(m)$-invariant differential forms on the sphere are the
  constants and the multiples of a volume form.

  If $m$ is even, the form $\beta$ restricts to a volume form on
  $S^{m-1}$ and belongs to $\mathcal C_\mu(X)$. If $m$ is odd there is
  still a unique rotation invariant volume form on $S^{m-1}$ up to
  normalization, its extension to a basic invariant form is given by
  the same formula as $\beta$, which is however not in
  $\mathcal C_\mu(X)$ due to the presence of the square root of $\mu$.
\end{proof}

\begin{lemma}\label{l-05}
  $H(\mathcal B(X))=0$.
\end{lemma}
\begin{proof}
  We imitate a standard proof of the Poincar\'e lemma. The dilation
  flow $\varphi_t(x)=tx$ for $t\in[0,1]$ maps balls centered at the
  origin to themselves. Let
  $h\colon \mathcal B(X)\to \mathcal B^{\bullet-1}(X)$ be the linear
  map
  \[
    h\,\omega=\int_0^1\varphi^*_t\iota_e\omega \, \frac{dt}{t},
  \]
  (it is well-defined since $\varphi^*_t\alpha$ vanishes for $t=0$ if
  $\alpha$ vanishes at the origin).  It is clear that $h$ maps forms
  vanishing to infinite order to forms vanishing to infinite
  order. Moreover, as in the proof of the Poincar\'e lemma, we see
  that $d\circ h+h\circ d=id-\varphi_0^*$.  But on forms vanishing at
  0, $\varphi^*_0$ vanishes. Thus the identity is homotopic to the
  zero map and the cohomology vanishes in all degrees.
\end{proof}

The long exact sequence associated with
\[
  0\to \mathcal B(X)\to\mathcal A_\mu(X)\to\mathcal A_\mu(X)/\mathcal
  B(X)\to0
\]
implies:
\begin{lemma}\label{l-06}
  $\mathcal H(A_\mu(X))\cong \mathcal H(\mathcal C_\mu(X)).$
\end{lemma}
\subsection{Filtration and the tame subcomplex}
Let $X$ be as above. The complex $\mathcal A_\mu(X)$ has a filtration
\[
  \cdots\subset F_p\mathcal A_\mu(X)\subset F_{p+1}\mathcal
  A_\mu(X)\subset\cdots\subset A_\mu(X) =\cup_{p\in\mathbb
    Z}F_pA_\mu(X)
\]
by subspaces
\[
  F_p\mathcal A_\mu(X) = \{ \omega\in\mathcal A_\mu(X) :
  \mu^{p}\omega, \mu^{p-1}d\mu\wedge\omega\in\mathcal A(X) \}.
\]
Since $\mathcal B(X)\subset \cap_{p\in\mathbb Z}F_p\mathcal A_\mu(X)$,
the filtration induces a filtration $F_p\mathcal C_\mu(X)$ of
$\mathcal C_\mu(X)=\mathcal A_\mu(X)/\mathcal B(X)$.
\begin{lemma}\label{l-007}
  Each $F_p\mathcal A_\mu(X)$ is a subcomplex preserved by $\iota_e$
  and invariant under $\mathrm{SO}(m)$. The same holds for the
  quotient complexes $F_p\mathcal C_\mu(X)$.
\end{lemma}
\begin{proof}
  The fact that $F_p\mathcal A_\mu(X)$ is a subcomplex is a special
  case of Lemma \ref{l-07}. As for the Euler vector field, we have
  \begin{align*}
    \mu^p\iota_e\omega
    &=\iota_e(\mu^p\omega),
    \\
    \mu^{p-1}d\mu\wedge \iota_e\omega
    &=-\iota_e(\mu^{p-1}d\mu\wedge\omega)+2\mu^p\omega.
  \end{align*}
  The right-hand sides are regular. Thus $F_p\mathcal A_\mu(X)$ is
  preserved by $\iota_e$. Since $\mu$ is rotation invariant, the
  action of $\mathrm{SO}(m)$ preserves the subcomplexes. Clearly
  $\mathcal B(X)$ is an $\mathrm{SO}(m)$-invariant subcomplex
  preserved by $\iota_e$, so the same holds for the quotient.
\end{proof}
Thus $F$ induces a filtration on $\mathrm{Ker}(L_e)$ and on
$\mathcal
C_{\mu,\mathrm{basic}}(X)=\mathrm{Ker}(L_e)\cap\mathrm{Ker}(\iota_e)$.
\begin{lemma}\label{l-008}
  The isomorphism of Lemma \ref{l-003} restricts to an isomorphism
  \[
    F_p\mathcal C_{\mu,\mathrm{basic}}(X)[-1]\oplus F_p\mathcal
    C_{\mu,\mathrm{basic}}(X)\to F_p\mathrm{Ker}(L_e)
  \]
  for all $p\in\mathbb Z$.
\end{lemma}
\begin{proof} It is easy to check that the filtration is preserved.
  By Lemma \ref{l-003} the map is injective. It remains to prove the
  surjectivity. Suppose $\omega\in F_p\mathrm{Ker}(L_e)$. Write
  $\omega=d\mu/\mu\wedge \sigma+\tau$ with
  $\sigma,\tau\in\mathcal C_{\mu,\mathrm{basic}}(X)$. Since
  $\mu^{p-1}d\mu\wedge\omega$ is regular, we deduce that
  $\mu^{p-1}d\mu\wedge\tau$ is regular on $X$. Applying $\iota_e$ and
  using that $\iota_e d\mu/\mu=2$ we see that $\mu^p\tau$ is also
  regular. Thus $\tau\in F_p\mathcal C_{\mu,\mathrm{basic}}(X)$. It
  follows that $\mu^{p-1}d\mu\wedge\sigma=\mu^{p}(\omega-\tau)$ is
  regular. Again applying $\iota_e$ we see that $\mu^p\sigma$ is
  regular, so that also $\sigma$ belongs to
  $F_p\mathcal C_{\mu,\mathrm{basic}}(X)$.
\end{proof}
\begin{cor}\label{c-02}
  Let the codimension $m$ of $Y$ in $X$ be even. Then the inclusion
  \[
    F_{\frac m2}\mathcal A_\mu(X)\hookrightarrow \mathcal A_\mu(X)
  \]
  is a quasi-isomorphism.
\end{cor}
\begin{proof}
  The differential forms $d\mu/\mu$, $\beta$, $d\mu/\mu\wedge\beta$
  are all in $F_{\frac m2}\mathcal A_\mu(X)$, which is preserved by
  multiplication by (pull-backs of forms in) $\mathcal A(Y)$. Thus in
  Lemma \ref{l-004} and Corollary \ref{c-01} we can replace the
  complexes by their $F_{\frac m2}$ subcomplexes.
\end{proof}

\end{document}